\documentclass[conference, a4paper]{IEEEtran}
\usepackage{cite}
\usepackage{graphicx, pstricks}
\usepackage{amsmath, amssymb, amsthm, pstricks, mathrsfs, subfig}
\newtheorem{theorem}{Theorem}

\newtheorem{proposition}{Proposition}

\newtheorem{remark}{Remark}

\newcommand{\bfset}{\mathpzc{x}}

\newcommand{\full}{\mathtt{full}}

\newcommand{\cerr}[2]{\mathrm{Q}(\sqrt{2|\langle #1, #2 \rangle|^2P})}

\DeclareFontFamily{OT1}{pzc}{}
\DeclareFontShape{OT1}{pzc}{m}{it}{<-> s * [1.10] pzcmi7t}{}
\DeclareMathAlphabet{\mathpzc}{OT1}{pzc}{m}{it}

\newcounter{cc}
\allowdisplaybreaks
\renewcommand{\qed}{\hfill \mbox{\raggedright \rule{0.1in}{0.1in}}}
\newcounter{consticnt}\addtocounter{cc}{1}
\date{}
\title{Variable-Length Channel Quantizers for Maximum Diversity and Array Gains\vspace{-4pt}}
\author{
\IEEEauthorblockN{Erdem Koyuncu and Hamid Jafarkhani}
\IEEEauthorblockA{Center for Pervasive Communications and Computing, University of California, Irvine.\vspace{-7pt}} 
}

\begin{document}
\maketitle
\begin{abstract}
We consider a $t \times 1$ multiple-antenna fading channel with quantized channel state information at the transmitter (CSIT). Our goal is to maximize the diversity and array gains that are associated with the symbol error rate (SER) performance of the system. It is well-known that for both beamforming and precoding strategies, finite-rate fixed-length quantizers (FLQs) cannot achieve the full-CSIT diversity and array gains. In this work, for any function $f(P)\in\omega(1)$, we construct variable-length quantizers (VLQs) that can achieve these full-CSIT gains with rates $1+(f(P) \log P)/P$ and $1+f(P)/P^t$ for the beamforming and precoding strategies, respectively, where $P$ is the power constraint of the transmitter. We also show that these rates are the best possible up to $o(1)$ multipliers in their $P$-dependent terms. In particular, although the full-CSIT SER is not achievable at any (even infinite) feedback rate, the full-CSIT diversity and array gains can be achieved with a feedback rate of $1$ bit per channel state asymptotically.
\end{abstract}

\section{Introduction}
The performance of a multiple antenna communication system can be greatly improved by making the channel state information (CSI) available at the transmitter and/or the receiver. Typically, the receiver can acquire the CSI by training. Obtaining CSI at the transmitter (CSIT) is however more difficult and generally requires receiver's feedback. In this context, a channel quantizer specifies (i) for each channel state, the sequence of feedback bits to be fed back by the receiver; and (ii) for each such sequence, the transmission codeword (e.g. a beamforming vector) to be employed by the transmitter. The goal is then to design an optimal quantizer with respect to a specific performance measure (such as the symbol error rate (SER)) subject to the rate constraint of the feedback link.

An overview on channel quantizers for multiple antenna systems can be found in \cite{love3}. In particular, limited feedback beamforming has been extensively studied via Grassmannian line packings\cite{love1, zhou2, raghavan1}, quantizer design algorithms\cite{narula1, roh1, roh2, lau1}, high resolution methods\cite{zheng1}, random vector quantizers\cite{yeung1, santipach2}, and several other techniques\cite{mukkavilli1, xia1, jafar1}. Beamforming is a special case of the more general, albeit more complex transmission strategy called precoding, whose performance with quantized feedback has been studied in \cite{jongren1, liu1, love4, love5, zhou1, ekbatani1}. 

Most of the previous work on CSI feedback employ fixed-length quantizers (FLQs), in which the number of feedback bits per channel state is a fixed integer. In general, different binary codewords of different lengths can be fed back for different channel states, resulting in a variable-length quantizer (VLQ). An FLQ is clearly a special case of a VLQ, and we thus expect VLQs to provide a better performance than FLQs. 

In this work, we focus on the VLQ design problem for a $t\times 1$ multiple-input single-output (MISO) system. Our goal is to maximize the diversity and array gains corresponding to the symbol error rate (SER) performance of the system (The reason why we will not focus on the SER itself will be explained later on.). We assume a quasi-static block fading channel model in which the channel realizations vary independently from one fading block to another while within each block they remain constant. We also assume that the receiver has full CSI, while the transmitter has only partial CSI provided by the receiver via error-free and delay-free feedback channels. The partial CSI is in the form of quantized instantaneous CSI provided by a VLQ. We design VLQs for beamforming and precoding strategies. For the latter strategy, we focus on linear precoding of complex orthogonal space-time block codes\cite{tjc}.

As we have shown previously, for a wide variety of communication problems, VLQs outperform FLQs by a significant margin\cite{dccfull, relayinterf, theotherdcc}. In \cite{dccfull}, we have considered the VLQ design problem for the outage probability performance measure. Compared to \cite{dccfull}, although the general idea behind designing good VLQs will remain the same, the specific methods of \cite{dccfull} are not directly applicable due to fundamentally different distortion functions. Also as a result of this difference, we obtain completely different results and reach different conclusions. 

In \cite{relayinterf, theotherdcc}, we have designed SER-optimizing distributed VLQs for beamforming in networks with multiple receivers. In this paper, the simpler non-distributed (point-to-point) nature of quantization allows us to prove much stronger achievability results. Here, we also prove converse results and study the more general precoding strategy. 

The rest of this paper is organized as follows: In Section \ref{secSystemModel}, we give a formal description of the system model. In Sections \ref{secBF} and \ref{secPC}, we state our main results for the beamforming and precoding strategies, respectively.

\emph{Notation:} For real-valued functions $f(x)$, $g(x)$, let $f(x)\in\omega(g(x))$ if (for all sufficiently large $x$) $f(x)\geq kg(x),\,\forall k>0$; $f(x)\in O(g(x))$ if $\exists a>0,\,f(x) \leq ag(x)$; $f(x) \in o(g(x))$ if $f(x) \leq \epsilon g(x),\,\forall \epsilon>0$; and $f(x)\sim g(x)$ if $\lim_{x\rightarrow\infty} \frac{f(x)}{g(x)} = 1$. $\mathbf{A}^T$, $\mathbf{A}^{\dagger}$ denote the transpose and the Hermitian transpose of matrix $\mathbf{A}$, respectively. $\mathbf{A}^* = (\mathbf{A}^T)^{\dagger}$. $\mathbf{h} \simeq \mathtt{CN}(\mathbf{K})$ means that $\mathbf{h}$ is a circulary-symmetric complex Gaussian random vector with covariance matrix $\mathbf{K}$. $\mathbf{I}$ is the identity matrix, and $\mathbf{0}$ is the all-zero matrix. $|\mathcal{A}|$ is the cardinality of $\mathcal{A}$. 

\section{Preliminaries}
\label{secSystemModel}
\subsection{System Model}
\label{secSystemModelSystemModel}
We consider a $t\times 1$ MISO system. Denote the channel from transmitter antenna $i$ to the receiver antenna by $h_i$, and let $\mathbf{h} = [\begin{array}{ccc} h_1 & \cdots & h_t \end{array}]^{T}\in\mathbb{C}^{t\times 1}$ represent the entire channel state. We assume that $\mathbf{h}\simeq\mathtt{CN}(\mathbf{I})$. 

We assume that the transmission symbol $s$ is a discrete random variable with a uniform distribution on the set $\{+1,-1\}$.\footnote{Our results can be extended to any finite constellation. We omit such a generalization here so as to highlight our quantization methods without dealing with the unnecessary technicalities of an arbitrary constellation.} For a fixed $\mathbf{h}$, we first consider the transmission of $s$ via a beamforming vector $\mathbf{x}\in\bfset$, where $\bfset = \{\mathbf{x}:\mathbf{x}\in\mathbb{C}^{t\times 1},\,\|\mathbf{x}\| = 1\}$ is the set of all feasible beamforming vectors (We shall discuss the precoded transmission later on.). The channel input-output relationship with such a tranmission strategy can be expressed as $y = s\langle \mathbf{x},\mathbf{h}\rangle \sqrt{P} + n$, where $y$ is the received signal, and the noise term $n\simeq\mathtt{CN}(1)$ is independent of $\mathbf{h}$. The corresponding signal-to-noise ratio (SNR) can be expressed as   $|\langle \mathbf{x},\mathbf{h} \rangle|^2 P$. Given $\mathbf{h}$, the conditional SER with a maximum-likelihood decoder is then $\mathrm{Q}(\sqrt{2 |\langle\mathbf{x},\mathbf{h}\rangle|^2P})$, where $\mathrm{Q}(x) = \frac{1}{\sqrt{2\pi}}\int_x^{\infty}e^{-\frac{u^2}{2}}\mathrm{d}\mathrm{x},\,x\in\mathbb{R}$ is the Gaussian tail function. 

When $\mathbf{h}$ is random, we can choose a different beamforming vector for different $\mathbf{h}$. In this case, we are interested in the SER averaged over all possible channel states. Formally, consider an arbitrary (measurable) mapping $\mathtt{m}:\mathbb{C}^t\rightarrow\bfset$. Then, the (average) SER with mapping $\mathtt{m}$ can be expressed as
\begin{align}
\label{serwithq}
\mathtt{SER}(\mathtt{m}) \triangleq \mathtt{E}[\mathrm{Q}(\sqrt{2 |\langle \mathtt{m}(\mathbf{h}),\mathbf{h} \rangle|^2 P})].
\end{align}
Let $\mathtt{d}(\mathtt{m}) = -\lim_{P\rightarrow\infty}[\mathtt{SER}(\mathtt{m})/\log P]$
as the diversity gain with $\mathtt{m}$, and $\mathtt{g}(\mathtt{m}) = \left[\lim_{P\rightarrow\infty}\left(\mathtt{SER}(\mathtt{m}) P^{\mathtt{d}(\mathtt{m})}\right)\right]^{-1}$
as the array gain with $\mathtt{m}$, provided that both limits exist. The asymptotic $P\rightarrow\infty$ performance of $\mathtt{m}$ is then $\mathtt{SER}(\mathtt{m}) \sim \left[\mathtt{g}(\mathtt{m}) P^{\mathtt{d}(\mathtt{m})}\right]^{-1}$. 

As an extreme case, the transmitter may know $\mathbf{h}$ perfectly, in which case we have a ``full-CSIT system.'' In such a scenario, we can choose an optimal beamforming vector, say $\full(\mathbf{h})$ for a given $\mathbf{h}$. We have $|\langle \full(\mathbf{h}),\mathbf{h} \rangle|  \leq \|\mathbf{h}\|$, and the upper bound is achievable by choosing $\full(\mathbf{h}) = \frac{\mathbf{h}}{\|\mathbf{h}\|}$. 
This gives us $\mathtt{SER}(\full) = \mathtt{E}[\mathrm{Q}(\sqrt{2 \|\mathbf{h} \|^2 P})]$ with $\mathtt{d}(\full) = t$. 

We now investigate the case where the transmitter has partial CSI via feedback from the receiver. Such a system can be modeled by a channel quantizer as we explain in the following.
\subsection{The channel quantizer}
\label{chquantdescription}
Let $\mathcal{I} \in \{\{0\},\{0,1\},\{0,1,2\},\ldots,\mathbb{N}\}$ be a possibly infinite index set. We use the notations $\{a_n\}_{\mathcal{I}}$ and $\{a_n:n\in\mathcal{I}\}$ interchangeably to represent a set whose elements are the real numbers $a_n,\,n\in\mathcal{I}$. A similar definition holds for sets of vectors, collection of sets, etc.  

Given $\mathcal{I}$, let $\{\mathbf{x}_n\}_{\mathcal{I}}$ be a set of quantized beamforming vectors with $\{\mathbf{x}_n\}_{{\mathcal{I}}}\subset\bfset$. Also, let $\{\mathcal{E}_n\}_{\mathcal{I}}$ with $\mathcal{E}_n\subset\mathbb{C}^t,\,\forall n\in\mathbb{N}$ be a collection of mutually disjoint measurable subsets of $\mathbb{C}^t$ with $\bigcup_{n\in\mathcal{I}} \mathcal{E}_n = \mathbb{C}^t$. Finally, let $\{\mathtt{b}_n\}_{\mathcal{I}}$ be a collection of feedback binary codewords with $\{\mathtt{b}_n\}_{\mathcal{I}} \subset \{0,1\}^{\star}$, where $\{0,1\}^{+} \triangleq \{\mathtt{0},\mathtt{1},\mathtt{00},\mathtt{01},\ldots\}$ is the set of all non-empty binary codewords. We assume that the code $\{\mathtt{b}_n\}_{\mathcal{I}}$ is prefix-free, which implies in particular that $\mathtt{b}_m \neq \mathtt{b}_n$ whenever $m \neq n$. We call the collection of triples $\mathtt{q}\triangleq \{\mathbf{x}_n,\,\mathcal{E}_n,\,\mathtt{b}_n\}_{\mathcal{I}}$ a quantizer $\mathtt{q}$ for the beamforming strategy. 

This definition immediately induces a feedback transmission scheme that operates in the following manner: For a fixed channel state $\mathbf{h}$, the receiver feeds back the binary codeword $\mathtt{b}_n$, where the index $n$ here satisfies $\mathbf{h}\in\mathcal{E}_n$. Such an index $n$ always exists and is unique as $\mathcal{E}_n,\,n\in\mathbb{N}$ is a disjoint covering of $\mathbb{C}^t$. The transmitter recovers the index $n$ and uses the corresponding beamforming vector $\mathbf{x}_n$. The recovery of $n$ by the transmitter is always possible since $\mathtt{b}_n$s are distinct. We write $\mathtt{q}(\mathbf{h}) = \mathbf{x}_n$ whenever $\mathbf{h}\in\mathcal{E}_n$ to emphasize the quantization operation. We call the set $\{\mathbf{x}_n\}_{\mathcal{I}}$ the quantizer (or beamforming) codebook. 

For any $\mathtt{b}\in\{0,1\}^{+}$, let $\mathtt{L}(\mathtt{b})$ denote the ``length'' of $\mathtt{b}$. For example, $\mathtt{L}(\mathtt{1}) = 1,\mathtt{L}(\mathtt{01}) = 2$. A quantizer $\mathtt{q}$ is called an FLQ if $\mathtt{L}(\mathtt{b}_m) = \mathtt{L}(\mathtt{b}_n),\,\forall m,n\in\mathcal{I}$. Otherwise, we call $\mathtt{q}$ a VLQ. In either case, the rate of $\mathtt{q}$ is $\mathtt{R}(\mathtt{q})\triangleq \sum_{n\in\mathcal{I}}\mathtt{P}(\mathbf{h}\in\mathcal{E}_n)\mathtt{L}(\mathtt{b}_n)$.

A quantizer $\mathtt{q}$ is thus a mapping $\mathbb{C}^t \rightarrow \{\mathbf{x}_n\}_{\mathcal{I}}$ supplied with a feedback binary codeword $\mathtt{b}_n$ for each $\mathbf{x}_n$. Treated solely as a mapping, it is a special case of the mapping $\mathtt{m}:\mathbb{C}^t\rightarrow\bfset$ discussed in Section \ref{secSystemModelSystemModel} with the requirement of a countable range $\{\mathbf{x}_n\}_{\mathcal{I}}$. We can therefore calculate the SER with $\mathtt{q}$ as $\mathtt{SER}(\mathtt{q}) = \mathtt{E}[\mathrm{Q}(\sqrt{2 |\langle \mathtt{q}(\mathbf{h}),\mathbf{h} \rangle|^2 P})]$. 

Our goal in this paper is to design low-rate VLQs that can achieve the full-CSIT diversity and array gains. Before we discuss our VLQ designs, let us reemphasize that we wish to achieve $\mathtt{d}(\full)$ and $\mathtt{g}(\full)$, not $\mathtt{SER}(\full)$ (although there is no practically-significant difference between the two goals at high $P$). The motivation behind this choice is not the hope that the diversity and array gains would be easier to work with. It is rather, as we shall prove in the following, the impossibility of achieving the full-CSIT SER with \textit{any} quantizer. Note that one may achieve $\mathtt{d}(\full)$ and $\mathtt{g}(\full)$ while not achieving $\mathtt{SER}(\full)$ at any $P$. For example, suppose that a hypothetical quantizer $\mathtt{q}'$ achieves $\mathtt{SER}(\mathtt{q}') = \mathtt{SER}(\full) + 1/P^{t+1}$. Obviously we have $\mathtt{SER}(\mathtt{q}') > \mathtt{SER}(\full),\,\forall P$, while $\mathtt{d}(\mathtt{q}') = \mathtt{d}(\full)$ and $\mathtt{g}(\mathtt{q}') = \mathtt{g}(\full)$.

\subsection{The impossibility of achieving $\mathtt{SER}(\full)$}
Let us first define the distortion function $d(\mathbf{x},\mathbf{h}) = \cerr{\mathbf{x}}{\mathbf{h}} - \mathrm{Q}(\sqrt{2\|\mathbf{h}\|^2P})$. For any given quantizer $\mathtt{q}$, we have $\mathtt{SER}(\mathtt{q}) = \mathtt{SER}(\full) + \mathtt{E}[d(\mathtt{q}(\mathbf{h}), \mathbf{h})]$, and therefore, minimizing the expected distortion with $\mathtt{q}$ is equivalent to minimizing the SER with $\mathtt{q}$. The following result is then merely the consequence of the countable nature of the quantizer structure and the properties of the distortion function $d(\mathbf{x},\mathbf{h})$ that is associated with the SER.
\begin{theorem}
\label{theorem1}
For\! any quantizer $\mathtt{q}$, $\mathtt{SER}(\mathtt{q}) \!\!>\! \mathtt{SER}(\full), \forall P\! >\! 0$. 
\end{theorem}
\begin{proof}
Let $\mathtt{q} = \{\mathbf{x}_n,\mathcal{E}_n,\mathtt{b}_n\}_{\mathcal{I}}$. Since $\sum_{n\in\mathcal{I}}\mathtt{P}(\mathbf{h}\in\mathcal{E}_n) = 1$, $\exists i\in\mathcal{I}$ such that $\mathtt{P}(\mathbf{h}\in\mathcal{E}_i)>0$. We then have
\begin{align}
\nonumber \textstyle \mathtt{SER}(\mathtt{q}) & \nonumber \textstyle = \mathtt{SER}(\full)  + \sum_{n\in\mathbb{N}} \int_{\mathcal{E}_n} d(\mathbf{x}_n,\mathbf{h})f(\mathbf{h})\mathrm{d}\mathbf{h} \\
& \geq \nonumber\textstyle \mathtt{SER}(\full) + \int_{\mathcal{E}_i}d(\mathbf{x}_i,\mathbf{h})f(\mathbf{h})\mathrm{d}\mathbf{h}.
\end{align}
For any $P > 0$, the distortion function $d(\mathbf{x}_i,\mathbf{h})$ is positive almost everywhere (it is non-negative everywhere and it is zero only when $\mathbf{h}$ is a member of the set $\{\mathbf{h}:|\langle \mathbf{x}_i,\mathbf{h}\rangle|^2 = \|\mathbf{h}\|^2\}$, which has probability measure zero). The integral of an almost-everywhere-positive function on a set of positive measure is positive. Hence, $\mathtt{SER}(\mathtt{q}) > \mathtt{SER}(\full),\forall P>0$.
\end{proof}
\begin{remark}[Comparison with outage probability\cite{dccfull}]
\normalfont For the same $t\times 1$ MISO system considered in this paper, and for a given target data transmission rate, let $\mathtt{OUT}(\mathtt{q})$ and $\mathtt{OUT}(\full)$ denote the outage probability with a (beamforming) quantizer $\mathtt{q}$ and the full-CSIT outage probability, respectively. Then, similarly, there is a distortion function $D(\mathbf{x},\mathbf{h})$ that satisfies $\mathtt{OUT}(\mathtt{q}) = \mathtt{OUT}(\full) + \mathtt{E}[D(\mathbf{x},\mathbf{h})]$. However, unlike the distortion function $d(\mathbf{x},\mathbf{h})$ for the SER, for any $\mathbf{x}\in\bfset$, the function $D(\mathbf{x},\mathbf{h})$ is zero on a set of positive measure\cite{dccfull}. Hence, the natural analogue of Theorem \ref{theorem1}, i.e. the claim that $\mathtt{OUT}(\mathtt{q})>\mathtt{OUT}(\full)$ for any quantizer $\mathtt{q}$, may not hold in the case of outage probability. In fact, it is possible to design finite-rate VLQs that can achieve $\mathtt{OUT}(\full)$ at any $P$ \cite{dccfull}. Hence, at least in the context of limited feedback, the SER and outage probability performance measures exhibit fundamentally different behaviors.\qed
\end{remark}

According to Theorem \ref{theorem1}, we have no hope in achieving $\mathtt{SER}(\full)$. The next most important question is how to design quantizers that can achieve $\mathtt{d}(\full)$ and $\mathtt{g}(\full)$ if possible. In this context, it is well-known that finite-rate FLQs cannot achieve these full-CSIT diversity and array gains \cite{roh1}. In this paper, we design VLQs that can achieve these gains with a feedback rate of $1$ bit per channel state asymptotically as $P\rightarrow\infty$. This is a significant improvement over FLQs that require infinite rate to achieve the same performance. We first discuss how to design such quantizers for the beamforming strategy; the precoding case will be discussed afterwards.

\section{VLQs for beamforming}
\label{secBF}
We start with the design of the quantizer encoding regions for a given beamforming codebook. We first recall the standard encoding rule for FLQs and discuss why it will not work in the case of VLQs. We then modify the standard encoder to come up with a new encoder that will allow us to design good VLQs.
\subsection{Encoding}
\label{secBFEncoding}
Let $\mathcal{B} = \{\mathbf{x}_n\}_{\mathcal{I}}$ be a finite-cardinality beamforming codebook. A standard practice (see e.g. \cite{mukkavilli1, roh1}) is to work with the quantizer 
$\mathtt{q}_{\mathcal{B}}(\mathbf{h}) \triangleq \arg\max_{\mathbf{x}\in\mathcal{B}} |\langle \mathbf{x}, \mathbf{h}\rangle|$, 
which chooses the beamforming vector (with ties broken arbitrarily) in $\mathcal{B}$ that is ``closest'' to $\mathbf{h}$. One way to design a VLQ might be to keep this standard encoding rule for FLQs but use a variable-length code instead of a fixed-length code. On the other hand, the standard encoding rule results in quantization cells with roughly equal probability $1/|\mathcal{B}|$ when $|\mathcal{B}|$ is large. In such a scenario, the rate of an optimal variable-length code will roughly be the same as the rate $\lceil \log_2\! |\mathcal{B}| \rceil$ of the fixed-length code\cite{dccfull}. We thus consider an alternate encoding strategy that can benefit from variable-length codes. 

The standard encoder always picks the best beamforming vector in $\mathcal{B} = \{\mathbf{x}_n\}_{\mathcal{I}}$ that is closest to $\mathbf{h}$. We do not have to be this precise if our goal is to achieve the diversity and array gains provided by $\mathcal{B}$. For example, we do not need to distinguish between two beamforming vectors given that both provide an SER of at most $o(1/P^t)$; preferring one vector over the other will not affect the diversity and array gains of the system as the best possible decay of the SER is $O(1/P^t)$.  

With this observation, for a given beamforming codebook $\mathcal{B}$, we consider a variable-length quantizer $\mathtt{q}_{\mathcal{B}}^{\mathrm{v}}$ that operates as follows. Let $\beta = (t+1)\log P$. 
\begin{itemize}
\item If $|\langle \mathbf{x}_i,\mathbf{h} \rangle|^2P \geq \beta,\,\forall i\in\mathcal{I}$, then $\mathtt{q}_{\mathcal{B}}^{\mathrm{v}}$ feeds back the binary codeword $\mathtt{0}$, and we set $\mathtt{q}_{\mathcal{B}}^{\mathrm{v}}(\mathbf{h}) = \mathbf{x}_0$.
\item Otherwise, if $\exists i\in\mathcal{I}$ with $|\langle \mathbf{x}_i,\mathbf{h} \rangle|^2P  < \beta$, then $\mathtt{q}_{\mathcal{B}}^{\mathrm{v}}$ feeds back the concatenation of the binary codeword $\mathtt{1}$ and the binary codeword of length $\lceil \log_2 |\mathcal{B}| \rceil$ bits that represents the index, say $j\in\mathcal{I}$, of the beamforming vector that results in the maximum SNR. We set  $\mathtt{q}_{\mathcal{B}}^{\mathrm{v}}(\mathbf{h}) =  \mathtt{q}_{\mathcal{B}}(\mathbf{h})$.
\end{itemize}

Let us now analyze the performance of $\mathtt{q}_{\mathcal{B}}^{\mathrm{v}}$.
\begin{proposition}
\label{firdtprop}
For any finite-cardinality beamforming codebook $\mathcal{B}$, we have 
\begin{align}
\label{serboundforqv}
\mathtt{SER}(\mathtt{q}_{\mathcal{B}}^{\mathrm{v}}) & \leq \mathtt{SER}(\mathtt{q}_{\mathcal{B}}) + \frac{1}{P^{t+1}}, \\
\label{rtboundforqv}
\mathtt{R}(\mathtt{q}_{\mathcal{B}}^{\mathrm{v}}) & \leq 1+\frac{(t+1)|\mathcal{B}|\log_2(4 |\mathcal{B}|) \log P}{P}.
\end{align}
\end{proposition}
\begin{proof}
We first prove the upper bound on $\mathtt{SER}(\mathtt{q}_{\mathcal{B}}^{\mathrm{v}})$. Let $d(\mathbf{h}) = \mathrm{Q}(\sqrt{2|\langle \mathtt{q}_{\mathcal{B}}^{\mathrm{v}}(\mathbf{h}), \mathbf{h}\rangle|^2P}) - \mathrm{Q}(\sqrt{2|\langle \mathtt{q}_{\mathcal{B}}(\mathbf{h}), \mathbf{h}\rangle|^2P})$, $\mathcal{E} = \{\mathbf{h}:\exists i\in\mathcal{I},\,|\langle\mathbf{x}_i,\mathbf{h}\rangle|^2P < \beta\}$, and $\mathcal{E}^c = \mathbb{C}^t - \mathcal{E}$. By the definition of $\mathtt{q}_{\mathcal{B}}^{\mathrm{v}}$, we have $\mathtt{E}[d(\mathbf{h}) | \mathbf{h}\in\mathcal{E}] = 0$.  Also, $\mathtt{E}[d(\mathbf{h}) | \mathbf{h}\in\mathcal{E}^c] \leq \mathtt{E}[\mathrm{Q}(\sqrt{2|\langle \mathtt{q}_{\mathcal{B}}^{\mathrm{v}}(\mathbf{h}), \mathbf{h}\rangle|^2P})] \leq \mathtt{E}[\mathrm{Q}(\sqrt{2(t+1)\log P})] \leq P^{-(t+1)}$, where the second inequality follows since $|\langle \mathtt{q}_{\mathcal{B}}^{\mathrm{v}}(\mathbf{h}), \mathbf{h}\rangle|^2P \geq (t+1)\log P,\,\forall\mathbf{h}\in\mathcal{E}^c$, and the last inequality follows from the bound $\mathrm{Q}(x) \leq e^{-x^2},\,x\geq 0$. Combining the two conditional expectations of $d(\mathbf{h})$, we obtain $\mathtt{E}[d(\mathbf{h})] \leq P^{-(t+1)}$. Substituting this to the obvious identity 
$\mathtt{SER}(\mathtt{q}_{\mathcal{B}}^{\mathrm{v}}) = \mathtt{SER}(\mathtt{q}_{\mathcal{B}}) + \mathtt{E}[d(\mathbf{h})]$, we obtain (\ref{serboundforqv}).

We now prove (\ref{rtboundforqv}). We have 
\begin{align}
\mathtt{R}(\mathtt{q}_{\mathcal{B}}^{\mathrm{v}}) & \textstyle = \mathtt{P}(\mathbf{h}\in\mathcal{E}^c) + (1+\lceil \log_2 |\mathcal{B}| \rceil) \mathtt{P}(\mathbf{h}\in\mathcal{E}) \\ & \textstyle \leq 1 +  \log_2 (4|\mathcal{B}|) \mathtt{P}(\mathbf{h}\in\mathcal{E}) \\ & \textstyle \leq 1+ \log_2 (4|\mathcal{B}|) \sum_{i\in\mathcal{I}} \mathtt{P}(|\langle \mathbf{x}_i,\mathbf{h}\rangle|^2 P \leq \beta) \\  & \textstyle = 1+ \log_2 (4|\mathcal{B}|) |\mathcal{B}| [1 - \exp(-\frac{\beta}{P})] \\ \textstyle & \textstyle \leq  1 + \log_2 (4|\mathcal{B}|) |\mathcal{B}| \frac{\beta}{P}, 
\end{align}
where the first inequality follows since $\mathtt{P}(\mathbf{h}\in\mathcal{E}^c) \leq 1$ and $1+\lceil \log_2 |\mathcal{B}| \rceil \leq 1+\log_2 |\mathcal{B}|+1 = \log_2(4|\mathcal{B}|)$. The second inequality follows from a union bound. This proves (\ref{rtboundforqv}). 
\end{proof}

Note that the maximum diversity gain with any quantizer is $t$. Hence, $\mathtt{SER}(\mathtt{q}_{\mathcal{B}}) \simeq g(\mathtt{q}_{\mathcal{B}})P^{-d}$ for some $d \leq t$. Since the second term in the upper bound in (\ref{serboundforqv}) decays faster than $\frac{1}{P^t}$, the diversity and array gains with $\mathtt{q}_{\mathcal{B}}^{\mathtt{v}}$ is the same with those of $\mathtt{q}_{\mathcal{B}}$. Moreover, according to (\ref{rtboundforqv}), we have $\mathtt{R}(\mathtt{q}_{\mathcal{B}}^{\mathtt{v}}) \rightarrow 1$ as $P\rightarrow\infty$. This is a significant improvement over a rate-$\lceil \log_2 |\mathcal{B}| \rceil$ FLQ for codebook $\mathcal{B}$, especially when $|\mathcal{B}|$ is large.

We now claim that, for any arbitrary function $f(P)\in\omega(1)$, there is a VLQ that can achieve $\mathtt{d}(\full)$ and $\mathtt{g}(\full)$ with rate $1+f(P)\frac{\log P}{P}$. We provide here an outline of the strategy to prove this result.  Motivated by (\ref{rtboundforqv}), we consider $P$-dependent codebooks $\mathcal{B}_P$ that satisfy $1+(t+1)|\mathcal{B}_P|\lceil \log_2 |\mathcal{B}_P| \rceil = f(P)$. In such a scenario, $|\mathcal{B}_P|\in\omega(1)$, or equivalently, we use codebooks with larger and larger cardinality as $P\rightarrow\infty$. If we can design the codebooks $\mathcal{B}_P,\,P>0$ well enough, the quantizer $\mathtt{q}_{\mathcal{B}_P}^{\mathtt{v}}$ will then achieve $\mathtt{d}(\full)$ and $\mathtt{g}(\full)$ with rate $1+f(P)\frac{\log P}{P}$, as claimed. Obviously, for this strategy to work, we need ``good'' codebook designs. We show the existence of such codebooks in the following.

\subsection{Existence of good codebooks}
The following result is a restatement of \cite[Proposition 2]{dccfull}.
\begin{proposition}
\label{zzprop}
For every sufficiently small $\delta > 0$, there is a codebook $\mathcal{B}_{\delta}$ with 
\begin{align}
\label{bdeltacardbound}
|\mathcal{B}_{\delta}| \leq \mathtt{C}_0\delta^{-2t}, 
\end{align}
and 
\begin{align}
\label{bdeltacovering}
\forall\overline{\mathbf{h}}\in\bfset,\,\exists \mathbf{x}\in\mathcal{B}_{\delta},\,|\langle\mathbf{x},\overline{\mathbf{h}}\rangle|^2 \geq 1-\delta, 
\end{align}
where $\mathtt{C}_0$ is a $\delta$-independent constant.
\end{proposition}
\begin{remark}
{\normalfont This existence result is good enough for our purposes. Still, let us note that an explicit construction for $\mathcal{B}_{\delta}$ is also available for implementation purposes. We refer the interested reader to \cite[Section III.B]{dccfull}.}\qed
\end{remark}

In the following proposition, we calculate the SER with $\mathcal{B}_{\delta}$.
\begin{proposition}
\label{relselprop}
For every sufficiently small $\delta > 0$, we have 
\begin{align}
\label{serqbdeltaubound}
\mathtt{SER}(\mathtt{q}_{\mathcal{B}_{\delta}}) \leq \mathtt{SER}(\full)(1 + 2t\delta). 
\end{align}
\end{proposition}
\begin{proof}
According to Proposition (\ref{zzprop}), the SNR provided by $\mathcal{B}_{\delta}$ is at least $\|\mathbf{h}\|^2(1-\delta)P$ for any given channel state $\mathbf{h}$. As a result, $\mathtt{SER}(\mathtt{q}_{\mathcal{B}_{\delta}}) \leq \int_0^{\infty} \mathrm{Q}(\sqrt{2x(1-\delta)P})x^{t-1}e^{-x}/\Gamma(t)\mathrm{d}x$, where the integration variable $x$ corresponds to a realization of the random variable $\|\mathbf{h}\|^2$. With a change of variables $u = x(1-\delta)$, we obtain  $\mathtt{SER}(\mathtt{q}_{\mathcal{B}_{\delta}}) \leq (1-\delta)^{-t}\int_0^{\infty} \mathrm{Q}(\sqrt{2uP})u^{t-1}e^{-u/(1-\delta)}/\Gamma(t)\mathrm{d}u$. For the factor $(1-\delta)^{-t}$ of the integral, we have $(1-\delta)^{-t} = 1+t\delta+O(\delta^2) \leq 1+2t\delta$ for any sufficiently small $\delta$. For the integral itself, we use the bound $e^{-\frac{u}{1-\delta}} \leq e^{-u}$, which makes the integral equal to $\mathtt{SER}(\full)$. This concludes the proof. 
\end{proof}
Hence, for sufficiently small $\delta$, the codebook $\mathcal{B}_{\delta}$ can provide the full-diversity gain. Making $\delta$ even smaller, it can also provide an array gain that is arbitrarily close to $\mathtt{g}(\full)$. 

\subsection{The main achievability result}
We can now proceed with the strategy outlined at the end of Section \ref{secBFEncoding}. The following is the main result of this section. 
\begin{theorem}
\label{bfachieve}
For any function $f(P)\in\omega(1)$, there is a quantizer $\mathtt{q}$ with $\mathtt{d}(\mathtt{q}) = \mathtt{d}(\full)$, $\mathtt{g}(\mathtt{q}) = \mathtt{g}(\full)$, and $\mathtt{R}(\mathtt{q}) \leq 1+f(P)\frac{\log P}{P}$ for all sufficiently large $P$. 
\end{theorem}
\begin{proof}
By (\ref{rtboundforqv}) and (\ref{bdeltacardbound}), for every sufficiently small $\delta$, we have $\mathtt{R}(\mathtt{q}_{\mathcal{B}_{\delta}})\leq 1+\phi(\delta)\frac{\log P}{P}$, where 
$\phi(\delta)\triangleq(t+1)\mathtt{C}_0/\delta^{2t}\times \break\log_2(4\mathtt{C}_0/\delta^{2t})$. Hence, for given $f(P)\in\omega(1)$, we choose $\delta(P) \triangleq \phi^{-1}(f(P))$. Since $f(P)\rightarrow\infty$, $\delta(P)$ is $o(1)$, well-defined, and positive for all sufficiently large $P$. The quantizer $\mathtt{q} = \mathtt{q}_{\mathcal{B}_{\delta(P)}}^{\mathtt{v}}$ now satisfies the theorem's statement. Indeed, by (\ref{serboundforqv}) and (\ref{serqbdeltaubound}), we have $\mathtt{SER}(\mathtt{q}) \leq \mathtt{SER}(\mathtt{Full})(1+2t\delta(P))+1/P^{t+1}$. Since $\delta(P)\in o(1)$, we have $\mathtt{d}(\mathtt{q}) = \mathtt{d}(\full),\mathtt{g}(\mathtt{q}) = \mathtt{g}(\full)$. The upper bound on $\mathtt{R}(\mathtt{q})$ is obvious.
\end{proof}
In particular, for any $f(P) \in o(\frac{P}{\log P})$ (e.g. $f(P) = \log P$), the full-CSIT diversity and array gains can be achieved with a feedback rate of $1$ bit per channel state asymptotically. The question is now to determine the minimum rate that guarantees the full-CSIT gains. We discuss this problem next.

\subsection{Necessary conditions for achieving $\mathtt{d}(\full)$ and $\mathtt{g}(\full)$}
It is difficult to determine the exact asymptotic rate that guarantees the full-CSIT gains. Instead, we provide bounds. Note that by Theorem \ref{bfachieve}, a quantization rate of $1+\frac{\omega(1)\log P}{P}$ is sufficient for the full-CSIT gains. We prove a partial converse by showing that a quantization rate of $1+\frac{t\log P}{13P}$ is necessary.
\begin{theorem}
\label{bfconversetheo}
For any quantizer $\mathtt{q}$, if $\mathtt{d}(\mathtt{q}) = \mathtt{d}(\full),\,\mathtt{g}(\mathtt{q}) = \mathtt{g}(\full)$, then $\mathtt{R}(\mathtt{q}) \geq 1+\frac{t\log P}{13P}$ for all sufficiently large $P$.
\end{theorem}
\begin{proof}
Let us first (roughly) identify the class of quantizers that can achieve $\mathtt{d}(\full)$ and $\mathtt{g}(\full)$.
Let $\mathtt{q} = \{\mathbf{x}_n, \mathcal{E}_n, \mathtt{b}_n \}_{\mathcal{I}}$. The case $|\mathcal{I}| = 1$ corresponds to open-loop system with no CSIT, in which case we have $\mathtt{d}(\mathtt{q}) \leq 1 < \mathtt{d}(\mathtt{full})$. The case $|\mathcal{I}| = 2$ corresponds to a quantizer whose codebook consists of $2$ beamforming vectors. In such a scenario, $\mathtt{d}(\mathtt{q}) = 2$ is achievable with $\mathtt{R}(\mathtt{q}) = 1$. On the other hand, for $t=2$, either $\mathtt{d}(\mathtt{q}) < \mathtt{d}(\full) = 2$, or if $\mathtt{d}(\mathtt{q}) = \mathtt{d}(\full) = 2$, then $\mathtt{g}(\mathtt{q}) < \mathtt{g}(\full)$ as one cannot achieve the full-CSIT array gain with only $2$ beamforming vectors. For $t\geq 3$, we simply have $\mathtt{d}(\mathtt{q}) \leq 2 < 3\leq \mathtt{d}(\full)$. 

It is thus sufficient to consider the case $|\mathcal{I}| \geq 3$ for the class of quantizers that can achieve $\mathtt{d}(\full)$ and $\mathtt{g}(\full)$. Hence, let $\mathtt{q} = \{\mathbf{x}_n, \mathcal{E}_n, \mathtt{b}_n \}_{\mathcal{I}}$ with $|\mathcal{I}|\geq 3$ and suppose that for every $P_0\in\mathbb{R}$, there exists $P \geq P_0$ such that $\mathtt{R}(\mathtt{q}) < 1+\frac{t\log P}{13P}$. We will conclude the proof of the theorem by showing that the strict inequality $\mathtt{d}(\mathtt{q}) < \mathtt{d}(\full)$ then holds.

Let $R = \frac{t\log P}{13P}$. Since, $|\mathcal{I}| \geq 3$, there is an index $i\in\mathcal{I}$ with $\mathtt{L}(\mathtt{b}_i) \geq 1$ and $\mathtt{L}(\mathtt{b}_j) \geq 2,\,\forall j\in\mathcal{I}-\{i\}$ (Otherwise, there would be two binary codewords of length $1$, which is impossible for a prefix-free code that contains at least $3$ codewords.). Then, we have $\mathtt{P}(\mathbf{h}\in\mathcal{E}_i) \geq 1-R$ (Otherwise, the chain of inequalities 
\begin{align}
\textstyle 1+R & > \mathtt{R}(\mathtt{q}) \\ & \textstyle = \sum_{n\in\mathcal{I}}\mathtt{P}(\mathbf{h}\in\mathcal{E}_n)\mathtt{L}(\mathtt{b}_n) \\ & \textstyle = \mathtt{P}(\mathbf{h}\in\mathcal{E}_i)\mathtt{L}(\mathtt{b}_i) + \sum_{j\in\mathcal{I}-\{i\}}\mathtt{P}(\mathbf{h}\in\mathcal{E}_j)\mathtt{L}(\mathtt{b}_j) \\ & \textstyle \geq \mathtt{P}(\mathbf{h}\in\mathcal{E}_i) + 2\sum_{j\in\mathcal{I}-\{i\}}\mathtt{P}(\mathbf{h}\in\mathcal{E}_j) \\ & \textstyle= \mathtt{P}(\mathbf{h}\in\mathcal{E}_i) + 2(1-\mathtt{P}(\mathbf{h}\in\mathcal{E}_i))\\ & \textstyle = 2-\mathtt{P}(\mathbf{h}\in\mathcal{E}_i) \\ & \textstyle\geq 1+R  
\end{align}
leads to a contradiction.). Without loss of generality, suppose that $\mathtt{P}(\mathbf{h}\in\mathcal{E}_0) \geq 1-R$. Then, with $f(\mathbf{h})$ representing the probability density function of $\mathbf{h}$, we have
\begin{align}
 \textstyle \mathtt{SER}(\mathtt{q}) & \textstyle = \sum_{n\in\mathcal{I}}\int_{\mathcal{E}_n} \mathrm{Q}(\sqrt{2|\langle \mathbf{x}_n,\mathbf{h}\rangle|^2 P})f(\mathbf{h})\mathrm{d}\mathbf{h} \\ &
 \textstyle \geq \int_{\mathcal{E}_0} \mathrm{Q}(\sqrt{2|\langle \mathbf{x}_0,\mathbf{h}\rangle|^2 P})f(\mathbf{h})\mathrm{d}\mathbf{h}
 \end{align}
 According to \cite[Theorem 2.1]{cote1}, there is a constant $\mathtt{C}_1>0$ such that $\mathrm{Q}(x) \geq \mathtt{C}_1 \exp(-x^2),\,\forall x\in\mathbb{R}$. Substituting this lower bound, we obtain
 \begin{align}
 \mathtt{SER}(\mathtt{q}) \geq g(\mathcal{E}_0,\mathbf{x}_0),
 \end{align}
 where 
 \begin{align}
 g(\mathcal{E},\mathbf{x}) = \textstyle \mathtt{C}_1 \int_{\mathcal{E}} \exp(-2|\langle \mathbf{x},\mathbf{h}\rangle|^2 P)f(\mathbf{h})\mathrm{d}\mathbf{h}
 \end{align}
Now, let $\mathfrak{E} = \{\mathcal{E}\subset\mathbb{C}^t:\mathtt{P}(\mathbf{h}\in\mathcal{E}) \geq 1-R\}$. Clearly, $\mathcal{E}_0\in\mathfrak{E}$, and therefore
\begin{align}
\mathtt{SER}(\mathtt{q}) \geq \inf_{\mathcal{E}\in\mathfrak{E}}\inf_{\mathbf{x}\in\bfset}g(\mathcal{E},\mathbf{x}).
\end{align}
Note that for any $t\times t$ unitary matrix $\mathbf{U}$, we have $g(\mathcal{E},\mathbf{x}) = g(\mathbf{U}\mathcal{E},\mathbf{U}\mathbf{x}),\,\forall\mathcal{E}\in\mathfrak{E},\,\forall\mathbf{x}\in\bfset$, where $\mathbf{U}\mathcal{E} = \{\mathbf{U}\mathbf{h}:\mathbf{h}\in\mathcal{E}\}$ denotes the translate of the set $\mathcal{E}$ by $\mathbf{U}$. With this property in mind, we consider a fixed vector $\mathbf{y}\in\bfset$. For a given $\mathbf{x}\in\bfset$, let the unitary matrix $\mathbf{U}_{\mathbf{x}}$ satisfy $\mathbf{x} = \mathbf{U}_{\mathbf{x}}\mathbf{y}$. We have
$g(\mathcal{E},\mathbf{x}) = g(\mathbf{U}_{\mathbf{x}}^{\dagger}\mathcal{E},\mathbf{U}_{\mathbf{x}}^{\dagger}\mathbf{x}) = g(\mathbf{U}_{\mathbf{x}}^{\dagger}\mathcal{E},\mathbf{y})$. Since $\mathbf{U}_{\mathbf{x}}^{\dagger}\mathcal{E}\in\mathfrak{E}$, we have $g(\mathcal{E},\mathbf{x})\geq \inf_{\mathcal{E}'\in\mathfrak{E}}g(\mathcal{E}',\mathbf{y})$. Since this inequality holds for arbitrary $\mathbf{x}$ and $\mathcal{E}$, we obtain $\inf_{\mathcal{E}\in\mathfrak{E}}\inf_{\mathbf{x}\in\bfset}g(\mathcal{E},\mathbf{x})\geq \inf_{\mathcal{E}\in\mathfrak{E}}g(\mathcal{E},\mathbf{y})$. Choosing e.g. $\mathbf{y} = [1\,0\,\cdots\,0]^T$ then gives us
\begin{align}
\textstyle \mathtt{SER}(\mathtt{q}) \geq \mathtt{C}_1\inf_{\mathcal{E}\in\mathfrak{E}} \int_{\mathcal{E}} \exp(-2|h_1|^2 P)f(\mathbf{h})\mathrm{d}\mathbf{h}
\end{align}
The expression in the lower bound is an optimization problem of the form ``Minimize $\int_{\mathcal{E}}\exp(-2|h_1|^2 P)\mathrm{d}\mu$, subject to $\mu(\mathcal{E}) \geq 1-R$,'' where $\mu$ is a probability measure. According to \cite{pkcwang1}, there is a minimizer of the form $\mathcal{E}' = \{\mathbf{h}: |h_1|^2 \geq r\}$, where $r$ is a positive real number with $\mu(\mathcal{E}') = 1-R$. In other words, one forms the solution set $\mathcal{E}'$ by starting with the points where the integrand $\exp(-2|h_1|^2 P)$ takes its minimal values and then progressively adds more points until the measure of $\mathcal{E}'$ is equal to $r$. We have
\begin{align}
\textstyle \int_{\{\mathbf{h}: |h_1|^2 \geq r\}}f(\mathbf{h})\mathrm{d}\mathbf{h} = 1-R,
\end{align}
or, equivalently $\int_r^{\infty}e^{-x}\mathrm{d}x = 1-R$. Solving for $r$, we obtain $r = -\log(1-R)$, and thus
\begin{align}
\textstyle \mathtt{SER}(\mathtt{q}) & \textstyle \geq \mathtt{C}_1\int_{-\log(1-R)}^{\infty} \exp(-2x P)e^{-x}\mathrm{d}x \\
& = \frac{\mathtt{C}_1\exp[(1+2P)\log(1-R)]}{1+2P} \\
& \geq \frac{\mathtt{C}_1\exp(-6PR)}{3P},
\end{align}
where the last inequality follows since $1+2P \leq 3P,\,\forall P \geq 1$, and $-\log(1-R) \leq 2R$ for sufficiently small $R$. Now, substituting $R = \frac{t\log P}{13P}$, the lower bound is proportional to $P^{-1 - \frac{6t}{13}}$. This means $d(\mathtt{q}) \leq 1 + \frac{6t}{13} < t,\,\forall t\geq 2$, if $\mathtt{d}(\mathtt{q})$ ever exists. This concludes the proof.  
\end{proof}
We thus say that the necessary and sufficient feedback rate that guarantees the full-CSIT gains is $1+\omega(1)\frac{\log P}{P}$, up to $o(1)$ multipliers in the $P$-dependent term $\omega(1)\frac{\log P}{P}$. The fundamental question of whether the rate of $1+\omega(1)\frac{\log P}{P}$ is necessary will remain an open problem.

\section{VLQs for Precoding}
\label{secPC}
So far, we have designed VLQs that can achieve the full-CSIT gains with rate $\smash[b]{1+\frac{\omega(1)\log P}{P}}$. The question is whether or not the decay rate of the term $\smash[b]{\frac{\omega(1)\log P}{P}}$ can be improved using more sophisticated data transmission strategies. The answer is yes, and as a proof of concept, we consider linear precoding of a complex orthogonal space-time block code.
\subsection{Linear precoding of complex orthogonal designs}
Consider a fixed $\mathbf{h}$. We consider the transmission of $k$ symbols $s_1,\ldots,s_k$ over $n$ time slots via an $n \times t$ complex orthogonal design $\mathbf{S}$ (with $\mathbf{S}^{\dagger}\mathbf{S}\!=\! \sum_k|s_k|^2 \mathbf{I},\, \forall s_1,\ldots,s_k\in\mathbb{C}$) multiplied by a precoding matrix $\mathbf{X}\in\mathcal{X}$, where $\mathcal{X} = \{\mathbf{X}\in\mathbb{C}^{t\times t}:\|\mathbf{X}\|\leq 1\}$. The channel input-output relationship can be expressed as $\mathbf{y} = \mathbf{S}\mathbf{X}\mathbf{h}\sqrt{P/r} + \mathbf{n}$, where $r = k/n$ is the space-time code rate, $\mathbf{y}$ is the received signal vector, and $\mathbf{n}\simeq\mathtt{CN}(\mathbf{I})$. Due to the orthogonality of the code, the per-symbol conditional SNR is $\|\mathbf{X}\mathbf{h}\|^2P/r$. It follows that the conditional SER (for BPSK) is given by $\mathrm{Q}(\sqrt{2\|\mathbf{X}\mathbf{h}\|^2P/r})$. In particular, if, for example, $\mathbf{X} = \mathbf{x}\mathbf{x}^{\dagger}$ for some $\mathbf{x}\in\bfset$, the conditional SER is $\mathrm{Q}(\sqrt{2|\langle\mathbf{x},\mathbf{h}\rangle|^2P/r})$. Hence, up to the $\frac{1}{r}$-scaling of the SNR, using the precoding matrix $\mathbf{x}\mathbf{x}^{\dagger}$ is equivalent to beamforming along $\mathbf{x}$.\footnote{In fact, using any precoding matrix of the form $\mathbf{z}\mathbf{x}^{\dagger},\,\mathbf{z}\in\bfset$ is equivalent to beamforming along $\mathbf{x}$. For example, $\mathbf{z} = [1\,\cdots\,1]^T$ gives us the matrix $[\mathbf{x}\,\cdots\,\mathbf{x}]^{\dagger}$. We have picked $\mathbf{z}=\mathbf{x}$ (which yields $\mathbf{x}\mathbf{x}^{\dagger}$) for a simpler notation.}

When $\mathbf{h}$ is random, we consider a mapping $\mathtt{M}:\mathbb{C}^{t\times 1}\rightarrow\mathcal{X}$, and let $\mathtt{SER}_r(\mathtt{M}) \triangleq \mathtt{E}[ \mathrm{Q}(\sqrt{2\|\mathtt{M}(\mathbf{h})\mathbf{h}\|^2P/r})]$ denote the SER with $\mathtt{M}$. The diversity and array gains $\mathtt{d}(\mathtt{M})$ and $\mathtt{g}(\mathtt{M})$ with mapping $\mathtt{M}$ can be defined in the same manner as in Section \ref{secSystemModelSystemModel}. The full-CSIT mapping $\mathtt{FULL}(\mathbf{h}) \triangleq \frac{\mathbf{h}\mathbf{h}^{\dagger}}{\|\mathbf{h}\|^2}$ provides the minimum possible SER, $\mathtt{SER}_r(\mathtt{FULL}) = \mathtt{E}[ \mathrm{Q}(\sqrt{2\|\mathbf{h}\|^2P/r})]$, by beamforming along $\frac{\mathbf{h}}{\|\mathbf{h}\|}$ for every $\mathbf{h}$.\footnote{For $t=2$ and the Alamouti code, we have $r=1$, and thus $\mathtt{SER}_1(\mathtt{FULL}) = \mathtt{SER}(\full)$. For $t \geq 3$, due to the non-existence of rate-$1$ complex-orthogonal designs, we have $\mathtt{SER}_r(\mathtt{FULL}) < \mathtt{SER}(\full)$: A ``part'' of what is lost in terms of data rate translates to a constant power gain.} 
\subsection{Main results}
The quantizer definition in Section \ref{chquantdescription} for beamforming extends to precoding in an obvious manner. We let the collection of triples $\mathtt{Q} \triangleq \{\mathbf{X}_n,\mathcal{E}_n,\mathtt{b}_n\}_{\mathcal{I}}$ be a channel quantizer for the precoding strategy, where $\{\mathbf{X}_n\}_{\mathcal{I}}\subset\mathcal{X}$ with $\{\mathcal{E}_n\}_{\mathcal{I}}$ and $\{\mathtt{b}_n\}_{\mathcal{I}}$ satisfying the properties described in Section \ref{chquantdescription}. 

It is straightforward to show that Theorem 1 holds for the precoding case as well. We thus construct quantized precoders that can achieve the full-CSIT diversity and array gains. We omit here a detailed exposition of the construction due to lack of space and since the ideas are very similar to the ones we have applied for beamforming. 

We consider the precoding codebook $\mathcal{C}_{\delta} \triangleq \{\mathbf{I}/\sqrt{t}\}\cup\{\mathbf{x}\mathbf{x}^{\dagger}:\mathbf{x}\in\mathcal{B}_{\delta}\}$ that consists of the identity precoder and the precoder equivalents of the beamforming vectors in $\mathcal{B}_{\delta}$. We then construct a VLQ $\mathtt{Q}_{\mathcal{C}_{\delta}}^{\mathtt{v}}$ that operates as follows:
\begin{itemize}
\item If $\|[\mathbf{I}/\sqrt{t}]\mathbf{h}\|^2 P \geq \delta^{-1}$, then $\mathtt{Q}_{\mathcal{C}}^{\mathtt{v}}$ feeds back the binary codeword $\mathtt{0}$, and we set $\mathtt{Q}_{\mathcal{C}_{\delta}}^{\mathtt{v}}(\mathbf{h}) = \mathbf{I}/\sqrt{t}$. 
\item Otherwise, $\mathtt{Q}_{\mathcal{C}_{\delta}}^{\mathtt{v}}$ feeds back the concatenation of the binary codeword $\mathtt{1}$ and the binary codeword of length $\lceil \log_2 |\mathcal{B}_{\delta}| \rceil$ bits that represents $\mathtt{q}_{\mathcal{B}_{\delta}}(\mathbf{h})$. We set $\mathtt{Q}_{\mathcal{C}_{\delta}}^{\mathtt{v}}(\mathbf{h}) = [\mathtt{q}_{\mathcal{B}_{\delta}}(\mathbf{h})][\mathtt{q}_{\mathcal{B}_{\delta}}(\mathbf{h})]^{\dagger}$.
\end{itemize}

In the following, we analyze the performance of $\mathtt{Q}_{\mathcal{C}_{\delta}}^{\mathtt{v}}$.
\begin{proposition}
For every sufficiently small $\delta$, 
\begin{align}
\label{serbndforpc}
\mathtt{SER}_r(\mathtt{Q}_{\mathcal{C}_{\delta}}^{\mathtt{v}}) & \leq \mathtt{SER}_r(\mathtt{FULL})(1+2t\delta) + \delta/P^t,\\
\label{rtbndforpc}
\mathtt{R}(\mathtt{Q}_{\mathcal{C}_{\delta}}^{\mathtt{v}}) & \leq 1+\mathtt{C}_2\delta^{-t}\log(\delta^{-1})/P^t,
\end{align}
where $\mathtt{C}_2$ is a constant that is independent of $\delta$ and $P$.
\end{proposition}
\begin{proof}
By construction, the precoding quantizer $\mathtt{Q}_{\mathcal{C}_{\delta}}^{\mathtt{v}}$ provides the same SNR as the beamforming quantizer $\mathtt{q}_{\mathcal{B}_{\delta}}$ except when $\|\mathbf{h}\|^2 P \geq t \delta^{-1}$, in which case it provides an SNR of $\|\mathbf{h}\|^2P/t$. This gives us
\begin{align}
\mathtt{SER}_r(\mathtt{Q}_{\mathcal{C}_{\delta}}^{\mathtt{v}}) & \leq \mathtt{SER}_r(\mathtt{q}_{\mathcal{B}_{\delta}}) + \int_{\frac{t}{\delta P}}^{\infty} \mathrm{Q}(\sqrt{2xP/t})\frac{x^{t-1}e^{-x}}{\Gamma(t)}\mathrm{d}x \\
& \leq \mathtt{SER}_r(\mathtt{q}_{\mathcal{B}_{\delta}}) + \int_{\frac{t}{\delta P}}^{\infty} \exp(-xP/t)\frac{x^{t-1}}{\Gamma(t)}\mathrm{d}x \\
& = \mathtt{SER}_r(\mathtt{q}_{\mathcal{B}_{\delta}}) + \frac{t^t}{P^t}\int_{\delta^{-1}}^{\infty} \frac{u^{t-1}e^{-u}}{\Gamma(t)}\mathrm{d}u \\
& = \mathtt{SER}_r(\mathtt{q}_{\mathcal{B}_{\delta}}) + \frac{t^t}{P^t}e^{-\delta^{-1}}\sum_{k=0}^{t-1} \frac{\delta^{-k}}{k!} \\
& \leq \mathtt{SER}_r(\mathtt{q}_{\mathcal{B}_{\delta}}) + \frac{\delta}{P^t},
\end{align}
where $\mathtt{SER}_r(\mathtt{q}_{\mathcal{B}_{\delta}}) \triangleq \mathtt{E}[\mathrm{Q}(\sqrt{2|\langle \mathtt{q}_{\mathcal{B}_{\delta}}(\mathbf{h}),\mathbf{h}\rangle|^2P/r})]$, and the last inequality holds for all sufficiently small $\delta$. Then, (\ref{serbndforpc}) follows since (using the same arguments as in Proposition \ref{relselprop}) $\mathtt{SER}_r(\mathtt{q}_{\mathcal{B}_{\delta}}) \leq \mathtt{SER}_r(\mathtt{FULL})(1+2t\delta)$.

For (\ref{rtbndforpc}), note that $\mathtt{Q}_{\mathcal{C}_{\delta}}^{\mathtt{v}}$ feeds back $1$ bit if $\|\mathbf{h}\|^2 P \geq t \delta^{-1}$, and it feeds back $1+\lceil \log_2 |\mathcal{B}_{\delta}| \rceil$ bits if $\|\mathbf{h}\|^2 P < t \delta^{-1}$. The latter event has probability 
\begin{align}
\int_0^{\frac{t}{\delta P}} \frac{x^{t-1}e^{-x}}{\Gamma(t)}\mathrm{d}x \leq \int_0^{\frac{t}{\delta P}} \frac{x^{t-1}}{\Gamma(t)}\mathrm{d}x = \frac{t^t\delta^{-t}}{\Gamma(t+1)P^t}.
\end{align}
Combining this with the bound $|\mathcal{B}_{\delta}| \leq \mathtt{C}_0\delta^{-2t}$ in (\ref{bdeltacardbound}), we obtain
\begin{align}
\mathtt{R}(\mathtt{Q}_{\mathcal{C}_{\delta}}^{\mathtt{v}}) \leq 1 + (1+\lceil \mathtt{C}_0\delta^{-2t} \rceil)\frac{t^t\delta^{-t}}{\Gamma(t+1)P^t},
\end{align}
After some straightforward manipulations, this yields the same upper bound as in the statement of the proposition.
\end{proof}
This immediately leads to the following theorem. The proof is omitted since it is very similar to the proof of Theorem \ref{bfachieve}.
\begin{theorem}
\label{pcachieve}
For any function $f(P)\in\omega(1)$, there is a precoding quantizer $\mathtt{Q}$ with $\mathtt{d}(\mathtt{Q}) = \mathtt{d}(\mathtt{FULL})$, $\mathtt{g}(\mathtt{Q}) = \mathtt{g}(\mathtt{FULL})$, and $\mathtt{R}(\mathtt{Q}) \leq 1+ \frac{f(P)}{P^t}$ for all sufficiently large $P$. 
\end{theorem}
For both beamforming and precoding strategies, we can achieve the full-CSIT diversity and array gains with a feedback rate of $1$ bit per channel state asymptotically as $P\rightarrow\infty$. By precoding, we can approach this $1$-bit rate asymptote with the much faster $\frac{\omega(1)}{P^t}$ decay compared to the $\frac{\omega(1)\log P}{P}$ decay with beamforming. We now show that the $1+\frac{\omega(1)}{P^t}$ decay is the best possible up to $o(1)$ multipliers in the term $\frac{\omega(1)}{P^t}$.
\begin{theorem}
There is a constant $\mathtt{C}_3>0$ such that for any quantizer $\mathtt{Q}$, if $\mathtt{d}(\mathtt{Q}) = \mathtt{d}(\mathtt{FULL})$ and $\mathtt{g}(\mathtt{Q}) = \mathtt{g}(\mathtt{FULL})$, then $\mathtt{R}(\mathtt{Q}) \geq 1+\frac{\mathtt{C}_3}{P^t}$ for all sufficiently large $P$.
\end{theorem}
\begin{proof}
Let $R = \frac{\mathtt{C}_3}{P^t}$ with the constant $\mathtt{C}_3$ to be specified later on. Also, let $\mathtt{Q} = \{\mathbf{X}_n,\mathcal{E}_n,\mathtt{b}_n\}_{\mathcal{I}}$ with $|\mathcal{I}|\geq 3$, and suppose that for every $P_0\in\mathbb{R}$, there exists $P \geq P_0$ such that $\mathtt{R}(\mathtt{Q}) < 1+R$. Using the same arguments as in the proof of Theorem \ref{bfconversetheo}, it is sufficient to show that $\mathtt{Q}$ will not be able to achieve the full-CSIT diversity and array gains. 

The properties of $\mathtt{Q}$ imply the existence of some $i\in\mathcal{I}$ with $\mathtt{P}(\mathbf{h}\in\mathcal{E}_i) \geq 1-R$ (see the proof of Theorem \ref{bfconversetheo}). With $f(\mathbf{h})$ denoting the probability density function of $\mathbf{h}$, we then have
\begin{align}
 \mathtt{SER}_r(\mathtt{Q}) &   = \sum_{n\in\mathcal{I}} \int_{\mathcal{E}_n}\mathrm{Q}(\sqrt{2\| \mathbf{X}_n\mathbf{h}\|^2P/r})f(\mathbf{h})\mathrm{d}\mathbf{h} \\
  &  \geq \int_{\mathcal{E}_i}\mathrm{Q}(\sqrt{2\|\mathbf{X}_i\mathbf{h}\|^2P/r})f(\mathbf{h})\mathrm{d}\mathbf{h}  \\
 &   = \int_{\mathbb{C}^T}\mathrm{Q}(\sqrt{2\|\mathbf{X}_i\mathbf{h}\|^2P/r})f(\mathbf{h})\mathrm{d}\mathbf{h} - \\ & \qquad  \int_{\mathcal{E}_1^c}\underbrace{\mathrm{Q}(\sqrt{2\|\mathbf{X}_i\mathbf{h}\|^2P/r})}_{\leq 1}f(\mathbf{h})\mathrm{d}\mathbf{h} \\
&  \geq \underbrace{\inf_{\mathbf{X}_i\in\mathcal{X}}\int_{\mathbb{C}^T}\mathrm{Q}(\sqrt{2\|\mathbf{X}_i\mathbf{h}\|^2P/r})f(\mathbf{h})\mathrm{d}\mathbf{h}}_{\triangleq \mathtt{SER}_r(\mathtt{OPEN})} - \\
& \qquad \int_{\mathcal{E}_1^c}f(\mathbf{h})\mathrm{d}\mathbf{h} \\
\label{lajldjasdas} &   \geq \mathtt{SER}_r(\mathtt{OPEN}) - R
\end{align}
where $\mathtt{SER}_r(\mathtt{OPEN})$ represents the open-loop SER with no feedback. Note that $\mathtt{d}(\mathtt{OPEN}) = t$ but $\mathtt{g}(\mathtt{OPEN}) < \mathtt{g}(\mathtt{FULL})$. We thus choose the constant in the statement of the theorem as $\mathtt{C}_3 = \frac{1}{2}(1/\mathtt{g}(\mathtt{OPEN}) - 1/\mathtt{g}(\mathtt{FULL}))$. Noting that $\mathtt{SER}_r(\mathtt{OPEN}) \sim 1/\mathtt{g}(\mathtt{OPEN})P^{-t}$ and $\mathtt{R}(\mathtt{q}) \leq \mathtt{\mathtt{C}_3}P^{-t}$, the lower bound in (\ref{lajldjasdas}) then becomes $\sim \frac{1}{2}(1/\mathtt{g}(\mathtt{OPEN}) + 1/\mathtt{g}(\mathtt{FULL}))P^{-t} > 1/\mathtt{g}(\mathtt{FULL})P^{-t}$. Hence, either $\mathtt{d}(\mathtt{Q}) < \mathtt{d}(\mathtt{FULL})  = t$, or if $\mathtt{d}(\mathtt{Q}) = t$, we have $\mathtt{g}(\mathtt{Q}) < \mathtt{g}(\mathtt{FULL})$. This concludes the proof.
\end{proof}
We can therefore conclude that for precoding, the necessary and sufficient feedback rate that guarantees the full-CSIT gains is $1+\frac{\omega(1)}{P^t}$, up to $o(1)$ multipliers in the term $\frac{\omega(1)}{P^t}$.

\section*{Acknowledgement}
This work was supported in part by the NSF Award CCF-1218771.

\end{document}